\documentclass{article}
\usepackage{arxiv}

\usepackage[utf8]{inputenc} 
\usepackage[T1]{fontenc}    
\usepackage{hyperref}       
\usepackage{url}            
\usepackage{booktabs}       
\usepackage{amsfonts}       
\usepackage{nicefrac}       
\usepackage{microtype}      
\usepackage{lipsum}		
\usepackage{graphicx}

\usepackage{amsmath,amssymb,amsthm}
\usepackage{mathtools}
\newcommand{\only}{\textit{only}}

\newcommand{\End}{\textsf{end}}

\newcommand{\Int}{\textit{int}}
\newcommand{\bool}{\textit{bool}}
\newcommand{\nat}{\textit{nat}}

\newcommand{\dom}{\textit{dom}}
\newcommand{\yes}{\textit{yes}}
\newcommand{\no}{\textit{no}}
\newcommand{\talk}{\textit{talk}}
\newcommand{\unsure}{\textit{unsure}}
\newcommand{\Next}{\textit{next}}
\newcommand{\quit}{\textit{quit}}

\newcommand{\erase}{\textit{erase}}
\newcommand{\Reach}{\textit{Reach}}
\newcommand{\prob}{\textit{prob}}
\newcommand{\Alice}{\textit{Alice}}
\newcommand{\Bob}{\textit{Bob}}

\newcommand{\System}{\textit{System}}
\newcommand{\PhoneNoA}{\textit{PhoneNoA}}

\newcommand{\IMIdA}{\textit{IMIdA}}

\newcommand{\ComPhone}{\textit{ComPhone}}
\newcommand{\ComIM}{\textit{ComIM}}
\newcommand{\noComm}{\textit{noComm}}
\newcommand{\nextProc}{\textit{nextProc}}

\newcommand{\E}{\mathcal{E}}
\newcommand{\0}{\mathbf{0}}

\def\defin#1#2{{\sf def}~ #1~ {\sf in} ~#2}

\newtheorem{example}{Example}
\newtheorem{definition}{Definition}
\newtheorem{theorem}{Theorem}
\newtheorem{lemma}{Lemma}

\title{Imprecise Probability \\ for Multiparty Session Types in Process Algebra}

\author{Bogdan Aman \\
Romanian Academy, Institute of Computer Science \\
Alexandru Ioan Cuza University, Ia\c si, Romania \\ 
\texttt{bogdan.aman@iit.academiaromana-is.ro} \\
\And
Gabriel Ciobanu \\
Romanian Academy, Institute of Computer Science \\
Alexandru Ioan Cuza University, Ia\c si, Romania \\ 
\texttt{gabriel@info.uaic.ro} }

\date{February 18, 2020}

\begin{document}

\maketitle

\begin{abstract}

In this paper we introduce imprecise probability for session types. 
More exactly, we use a probabilistic process calculus in which both 
nondeterministic external choice and probabilistic internal choice are 
considered. We propose the probabilistic multiparty session types able to 
codify the structure of the communications by using some imprecise 
probabilities given in terms of lower and upper probabilities. We prove 
that this new probabilistic typing system is sound, as well as several 
other results dealing with both classical and probabilistic properties. 
The approach is illustrated by a simple example inspired by survey polls.

\end{abstract}


\section{Introduction}

Aiming to represent the available knowledge more accurately, imprecise 
probability generalizes probability theory to allow for partial probability 
specifications applicable when a unique probability distribution is 
difficult to be identified. The idea that probability judgments may be 
imprecise is not novel. It would be enough to mention the Dempster-Shafer 
theory for reasoning with uncertainty by using upper and lower 
probabilities~\cite{Dempster67}.
The idea appeared even earlier in 1921, when J.M. Keynes used explicit 
intervals for approaching probabilities~\cite{Keynes21} (B. Russell called 
this book ``undoubtedly the most important work on probability that has 
appeared for a very long time''). In analytic philosophy, the probability 
judgements were formulated in terms of interval-valued functions 
in~\cite{Levi74}, and in terms of upper and lower probabilities 
in~\cite{Wolfenson82}. The term {\it imprecise probability} was introduced 
explicitly in 1991 \cite{Walley91}, while in 2000 the theory of interval 
probability was presented as a unifying concept for uncertainty 
\cite{Weich00}. Information on how imprecise probability differs from the 
classic Bayesian approach could be found in \cite{Halpern03}. 

Using imprecise probability in the framework of computer science 
involving process calculi and session types represents the contribution of 
this paper. More exactly, imprecise probability are used to represent and 
quantify uncertainty in the study of concurrent processes involving 
multiparty session types. We consider a probabilistic process calculus 
using nondeterministic branching (choices made by an external process) and 
probabilistic selection (choices made internally by the process). 
This new probabilistic approach satisfies the main standard axioms for 
determining the probability of a~event\footnote{The probability of 
an event is a real number in the interval $[0,1]$, and the sum of the 
probabilities of the elementary events is $1$; these axioms 
represent a simplified version of those introduced in 1933 by A.Kolmogorov 
in the context of measure theory~\cite{Kolmogorov50}.}.

Session types \cite{Honda93} and multiparty session types~\cite{Honda16} 
describe a type discipline for communication-centric systems 
having their roots in process calculi. Essentially, session types provide a 
typing discipline ensuring that a message-passing process implements a 
given (multiparty) session protocol defined as a structured sequence of 
interactions without errors. Such a structure is abstracted as a type 
through an intuitive syntax which is used to validate programs. Session 
types are terms of a process calculus that also contains a selection 
construct (an internal choice among a set of branches), a branching 
construct (an external choice offered to the environment) and recursion.

To understand deeper the quantitative aspects of uncertainty that might 
arise in multiparty session types, we use {\it imprecise probabilities} by 
considering sets of probabilities (rather than a single probability) in 
terms of lower and upper probabilities. The imprecise probabilities appear 
for instance in voting polls, where it is reported an interval given by a 
percentage and a margin error. The probability interval is determined by 
subtracting and adding the sampling (margin) error to the sample mean 
(estimated percentage). This confidence interval may be wider or narrower, 
depending on the degree of (un)certainty or a required precision. 
To illustrate our approach, we consider a simple example inspired by 
probabilistic survey polls (which represent an alternative to the classic 
survey polls). In such a probabilistic poll the pollsters 
communicate their own beliefs and receive the beliefs of others in the 
form of numerical probabilities~\cite{Wallsten93}. One of the advantages of 
such an approach is that the probability provides a well-defined numerical 
scale for responses; in \cite{Manski} it is studied the accuracy of voting 
probabilities in the US presidential election (reporting the first 
large-scale application of probabilistic polling).
The concept of probabilistic polling is based on the assumption that the 
opinions of the people accurately represent the distribution of opinions 
across the entire population; however, this can never be completely true. 
The margin of error describes the uncertainty that comes from having a 
small sample size (relative to the size of the polling population).

In the following, Section \ref{section:multiparty_processes} presents the 
syntax and semantics of our calculus; the key ideas are explained by using 
a running example. Section~\ref{section:multiparty_session} describes both 
global and local types used in session types, as well as the connection 
between them. Section~\ref{section:typing_system} presents the new 
probabilistic typing system, and the main results. The probabilistic 
extension preserves the classical properties of a typing system; 
additionally, it satisfies the axioms of the probability theory, 
and can determine the imprecise probability of certain behaviours. 
Section~\ref{section:conclusion} concludes and discusses certain related 
probabilistic approaches involving typing~systems.

\section{Processes in Probabilistic Multiparty Sessions}
\label{section:multiparty_processes}

In this section we describe a probabilistic extension of the multiparty 
session process calculus presented in~\cite{Scalas17}, by using the same 
notations as in~\cite{Scalas17}. This probabilistic extension allows 
probabilistic selection made internally by the communicating processes, 
and also branching controlled by an external process.

\subsection{Syntax}

Informally, a session represents a unit of conversation given by a series 
of interactions between multiple parties. A session is established via a 
shared name representing a public interaction point, and consists of 
communication actions performed on fresh session channels. The syntax is 
presented in Table \ref{table:syntax}, where the following notations are 
used: probabilities $p$, channels~$c$ and~$s$, values $v$, variables~$x$, 
labels $l$, role $r$ and process variables $X$ in order to define recursive 
behaviours. In what follows, $fc(P)$ denotes the set of free channels with 
roles in $P$, while~$fv(P)$ denotes the set of free variables in~$P$.

\begin{table}[h!]
  \centering
  \begin{tabular}{l@{\hspace{1ex}}lll@{\hspace{2ex}}r}
   \hline
   \vspace{-1.5ex}
$~$\\
\vspace{0.25ex}
    {\it Processes} & $P,Q$ & $::=$ &  $ c[r] \oplus_{i\in I}\langle p_i:l_i(v_i); P_i\rangle$   & (selection towards role $r$, $I \neq \emptyset$)\\
    \vspace{0.25ex}
    & & \quad $\shortmid$ & $ c[r]\ \&_{i\in I}\{l_i(x_i); P_i\}$   & (branching from role $r$, $I \neq \emptyset$)\\
        \vspace{0.25ex}
    & & \quad $\shortmid$ & $(\nu s)P$   & (restriction)\\   
        \vspace{0.25ex}
    & & \quad $\shortmid$ & $\defin{D}{P}$   & (process definition)\\
        \vspace{0.25ex}
    & & \quad $\shortmid$ & $X(\tilde{v})$   & (process call)\\
        \vspace{0.25ex}
    & & \quad $\shortmid$ & $\0$   & (inaction)\\
     \vspace{0.25ex}
    & & \quad $\shortmid$ & $P \mid Q$   & (parallel)\\
        \vspace{0.25ex}
    {\it Declaration} & $D$ & $::=$ & $X(\tilde{x})=P$& (process declaration) \\  
        \vspace{0.25ex}
    {\it Context} & $\E$ & $::=$ & $[~] \mid P$& (evaluation context) \\  
        \vspace{0.25ex}        
         & & \quad $\shortmid$ & $(\nu s)\E$  \\    
                 \vspace{0.25ex}        
         & & \quad $\shortmid$ & $\defin{D}{\E}$  \\   
                          \vspace{0.25ex}        
         & & \quad $\shortmid$ & $\E \mid \E$  \\   
    {\it Channel} & $c$ & $::=$ & $x$ $\shortmid s[r]$  & (variable, channel with role $r$) \\  
    \vspace{0.25ex}
    {\it Values} & $v$ & $::=$ & $s \shortmid  true \shortmid false$ $\shortmid 3 \shortmid \ldots$ & (base value)\\  
    \hline 
    \vspace{0.15ex}
   \end{tabular}
\caption{Syntax of the probabilistic multiparty session calculus}\label{table:syntax}
\end{table}

In both selection and branching, the labels $l_i$ ($i \in I$) are all 
different and their order is irrelevant. This ensures that a labelled 
choice in a selection can be matched uniquely by its corresponding label 
in a branching. The process $ c[r] \oplus_{i\in I}\langle 
p_i:l_i(v_i); P_i\rangle$ performs an internal choice on the channel~$c$ 
towards role $r$, the labelled value $l_i(v_i)$ is sent with the 
probability~$p_i$, and the execution continues as process $P_i$. Depending 
on which label is selected, there are~$|I|$ possibilities of continuation, 
all unique due to the uniqueness of the labels in each set. When all the 
sets of probabilistic choices of a process $P$ contain only one element 
with probability~$1$, then the process $P$ can be defined using the syntax 
presented in~\cite{Scalas17} (by discarding all the probabilities). On the other  
hand, the process $ c[r]\&_{i\in I}\{l_i(x_i); P_i\}$ waits 
for an external choice on the channel $c$ from role $r$. If the labelled 
$l_i$ is used to communicate, then the execution continues as $P_i$ in 
which all the occurrences of the variable $x_i$ are replaced by the received 
value. Note that for all $i \in I$, the variables $x_i$ are bound with 
scope $P_i$. The restriction~$(\nu s)P$ delimits the scope of a channel 
$s$ to $P$. Process definition $\defin{D}{P}$ and process call~
$X(\tilde{v})$ describe recursion, with $D$ being a process declaration of 
the form~$X(\tilde{x})=P$. The call invokes $X$ by replacing its formal 
parameters $\tilde{x}$ with the actual ones $\tilde{v}$ from the process 
call. Just like in \cite{Scalas17}, we consider closed process 
declarations for which we have $fv(P) \subseteq \tilde{x}$ and $fc(P) = 
\emptyset$ whenever $X(\tilde{x})=P$. The inaction ${\bf 0}$ represents a 
terminated process, while the parallel composition $P\ |\ Q$ represents 
two processes that execute concurrently and possibly communicate on a 
shared channel. As usual, evaluation contexts are processes with some 
holes~\cite{Coppo16}. A channel $c$ can be either a variable~$x$ or a 
channel with role $s[r]$, namely a multiparty communication channel for 
role $r$ in a session $s$. Values~$v$ can be integers, Boolean and
strings (i.e., base values).

\begin{example}\label{example:processes}
The following process describes a simple survey poll in which $\Alice$ 
either responds to some questions of $\Bob$ or declines the polling.  
$\Bob$ can either stop the polling if it is pleased with the responses 
received till now, or has some other polling to perform. Just like 
in~\cite{Wallsten93}, $\Alice$ and $\Bob$ communicate their own 
choices and receive the choices of others having numerical 
probabilities attached to them, probabilities providing a well-defined 
numerical scale for~responses.
\smallskip

\centerline{$\System = (\nu s)(\;\Bob\ |\ (\nu \PhoneNoA,\IMIdA)\Alice\;)$ .}
\smallskip

\noindent The name $s$ indicates the session involving the processes 
$\Alice$ and $\Bob$ composed in parallel. Session $s$ has two roles:~$r_B$ 
for the process performing the survey poll, and $r_A$ for the process 
answering the poll. The processes $\Alice$ and $\Bob$ 
communicate in this session by using the channel with role $s[r_A]$ and 
$s[r_B]$, respectively. In a similar manner, $\PhoneNoA$ and $\IMIdA$ 
represent sessions names used to model sessions established 
between $\Alice$ and $Bob$ after the communication performed inside the 
session $s$ ends. Note that the scopes of $\PhoneNoA$ and $\IMIdA$ include 
only the process $\Alice$ initially, as the process $\Bob$ is informed 
about which session to join after communication with $\Alice$.
\smallskip

$\Alice= {\sf def}~A(y)= A_1 ~{\sf in} ~ ~A_2$, 
where $A_1$ and $A_2$ are the following processes:

\hspace{5ex}$A_1 = y[r_B]\& \{\talk(t_1). y[r_B] \oplus \langle 0.6:\yes(t).A(y)~,~0.3: \no(t).A(y)~,~0.1:\quit(t).\0 \rangle)~,~\quit(t_2).\0\}$ ;

\hspace{5ex}$A_2 = s[r_A][r_B]\oplus \langle 0.6 :\ComPhone(\PhoneNoA).A(\PhoneNoA[r_A])~,$

\hspace{22ex}$~0.35: \ComIM(IMIdA).A(\IMIdA[r_A])~,$

\hspace{22ex}$~0.05:\noComm(\no).\0 \rangle$ .

$\Bob= {\sf def}~B(y,t)= B_1 ~{\sf in} ~ ~B_2$, 
where $B_1$ and $B_2$ are the following processes:

\hspace{5ex}$B_1=y[r_A] \oplus \langle 0.95:\talk(t).y[r_A] \& \{ \yes(t).B(y,\Next(t))~,~\no(t).B(y,\Next(t))~,$

\hspace{40ex}$~\unsure(t).B(y,\Next(t))~, ~\quit(t).B \} ~,$

\hspace{19ex}$~0.05:\quit(\no).\0\rangle$ ;

\hspace{5ex}$B_2=s[r_B][r_A]\ \& \ \{\ComPhone(x_A).B(x_A[r_B],Q_1)~,~ \ComIM(x'_A).B(x'_A[r_B],Q_1)~,~\noComm(x''_A).\0\}$ .

\smallskip

We explain the concurrent evolution of processes~$\Alice$ and $\Bob$ by 
describing the interaction between them. In the first step, using the 
channel with role $s[r_A]$, the process $\Alice$ decides whether to respond 
to the poll either by phone, instant messaging, or not answering now. In 
each case, process $\Alice$ has attached the probabilities marking its 
willingness to take the appropriate branch. Suppose the $\ComPhone$ label 
is used; in this case $\Alice$ and $\Bob$ communicate by using the session 
channel $\PhoneNoA$ that represents the unique phone number of $\Alice$. 
The process $\Bob$ is implemented by invoking $B(x_A[r_B])$, where $x_A$ 
becomes $\PhoneNoA$ after communication. Here $Q_1$ stands for the first 
question that process $A$ has to respond by using one of the labels $\yes$, 
$\no$, $\unsure$ or $\quit$. Notice that the process $\Alice$ does not 
provide the answer by using $\unsure$, as does not consider this to be an 
option for the given question. After receiving an answer to the question, 
process $\Bob$ can continue to ask the next question from its list 
(illustrated by $\Next$). The processes can quit the poll by performing a 
selection of either the label $\quit$ or $\noComm$ (when appropriate). 
Notice that the $\quit$ and $\noComm$ have low attached probabilities; this 
means that these options have very low chances of being chosen when the 
survey poll is conducted. 
\end{example}
 
\subsection{Operational Semantics}

In what follows, $s \not\in fc(P)$ means that it does not exist an $r$ such 
that $s[r] \in fc(P)$. Also, $dpv(D)$ denotes the set of process variables 
declared in~$D$, while $fpv(P)$ denotes the set of process variables which 
occur free in $P$. 

The operational semantics is based on the notion of structural congruence:
 
\vspace{-2ex}
\begin{center}
\begin{tabular}{c}
$P | Q \equiv Q | P \qquad (P | Q) | R \equiv P | (Q | R) \qquad P | \0 \equiv P \qquad (\nu s)\0 \equiv \0$\\
        \vspace{0.25ex}
$(\nu s)(\nu s')P \equiv (\nu s')(\nu s)P \qquad (\nu s)P | Q \equiv (\nu s)(P | Q)~ ({\rm if}~ s \not\in fc(Q))$\\
        \vspace{0.25ex}
$\defin{D}{\0} \equiv \0 \qquad \defin{D}{(\nu s)P} \equiv (\nu s)\defin{D}{P} ~({\rm if}~ s \not\in fc(D))$\\
        \vspace{0.25ex}
$\defin{D}{(P | Q)} \equiv (\defin{D}{P}) | Q~ ({\rm if}~ dpv(D) \cap fpv(Q) = \emptyset)$\\
$\defin{D}{\defin{D'}{P}} \equiv \defin{D'}{\defin{D}{P}}$\\
        \vspace{0.25ex}
$({\rm if}~ (dpv(D) \cup fpv(D)) \cap dpv(D') = (dpv(D') \cup fpv(D')) \cap dpv(D) = \emptyset)$.
\end{tabular}
\vspace{-1.5ex}
\end{center}

The operational semantics provides the opportunity of uniquely 
identifying each individual transition. To do this, we use the transition 
labels attached to each execution step. The semantics of our probabilistic 
calculus is presented in~Table~\ref{table:semantics}.
\vspace{-0.5ex}

\begin{table}[ht]
  \centering
  \begin{tabular}{@{\hspace{0ex}}l@{\hspace{1ex}}c@{\hspace{0ex}}}
   \hline
\vspace{1ex}

$ s[r_1][r_2] \oplus_{j\in J}\langle p_j:l_j(v_j); P_j\rangle  
\mid s[r_2][r_1]\&_{i\in I}\{l_i(x_i); P'_i\} \xrightarrow{(r_1,r_2,l_k)}_{p_k} 
P_k \mid  P'_k\{\tilde{v_k}/\tilde{x_k}\}$ & ({\sc Com}) \\
\vspace{1.25ex}
\hfill(if  $k \in J \subseteq I$) \\

\vspace{1.25ex}
${\sf def}~ X(\tilde{x}) = P ~{\sf in}~ (X(\tilde{v})\ |\ Q) \xrightarrow{\varepsilon}_1 {\sf def}~ X(\tilde{x}) = P ~{\sf in}~ (P\{\tilde{v}/\tilde{x}\}\ |\ Q)$ & ({\sc Call})\\
\vspace{0.5ex}
\hfill(if  $\tilde{x}=x_1\ldots x_n$, $\tilde{v}=v_1\ldots v_n$) \\

\vspace{1.25ex} 
$P \xrightarrow{tl}_p P'$ implies $\E[P] \xrightarrow{tl}_{p'} \E[P']$ (where  $p'=p \cdot \nextProc(P)/ \nextProc(\E(P))$) & ({\sc Ctxt})\\
\vspace{0.5ex}
$P\equiv P'$ and $P' \xrightarrow{tl}_r Q'$ and $Q \equiv Q' $ \ implies \ $ P \xrightarrow{tl}_r Q$ & ({\sc Struct})\\
\hline 
\vspace{0.15ex}
\end{tabular}
\caption{Operational semantics of the probabilistic multiparty session calculus}
\label{table:semantics}
\end{table}

\vspace{-2ex}
Rule {\sc (Comm)} models the communication between two roles $p$ and $q$ 
by using a branching process and a probabilistic selection process. If the 
process $ s[r_1][r_2] \oplus_{j\in J}\langle p_j:l_j(v_j); P_j\rangle $ 
chooses a branch with label $l_k$ and probability $p_k$, then the 
corresponding branch with label~$l_k$ is chosen also in process $ 
s[r_2][r_1]\&_{i\in I}\{l_i(x_i); P'_i\}$. To identify which branch was 
selected, we use the transition label $(r_1,r_2,l_k)$. The first process 
continues as $P_k$, while the second one as $P'_k\{\tilde{v_k}/\tilde{x_k}\}$ 
obtained by using the communicated value $\tilde{v_k}$ and substituting it 
for the existing variable~$\tilde{x_k}$. It is worth noting that it is 
possible that the options offered by the selection process are fewer than 
the ones offered by the branching process because the last one has to take into 
account all the possible continuations (depending on different scenarios). 
However, all the options of the selection set $J$ should be present in the 
branching set $I$ (expressed by the condition $J \subseteq I$) in order 
to avoid choosing a branch with no correspondence.

Rule {\sc (Call)} instantiates a process call by using its definition and 
replacing its formal parameters $\tilde{x}$ with the actual ones~$\tilde{v}$. 
The side condition of this rule ensures that the number of variables of 
the definition is equal with the number of variables appearing in the 
process call. Since the continuation is unique, we should not mention what 
it was executed, and so it is used the transition label $\varepsilon$. 

Rule ({\sc Struct}) states that the transition relation is closed under 
structural~congruence.

Rule {\sc (Ctxt)} states that the transition can happen under restriction, 
process definition and parallel composition; the transition label $tl$ 
stands for either a tuple of the form $(r_1,r_2,l_j)$ or~$\varepsilon$. The new value $p'$ of the probability attached to the reduction arrow is a normalization of the old probability value $p$ with respect to the number of rules {\sc (Comm)}  and {\sc (Call)} that can be applied. To compute the number of possible rules {\sc (Comm)}  and {\sc (Call)} that can be applied to a process $P$, we use the $\nextProc$ function defined as follows:

\centerline{$\nextProc(P) = 
\begin{cases}
1+\nextProc(P') & \mbox{if } P = s[r_1][r_2] \oplus_{j\in J}\langle p_j:l_j(v_j); P_j\rangle \\
& \hspace{5.5ex}\mid s[r_2][r_1]\&_{i\in I}\{l_i(x_i); P'_i\} \mid P'\\
\nextProc(P')+\nextProc(P'') & \mbox{if } P = ((\nu s')P') \mid P'''\\
1+\nextProc((\defin{D}{P'})\mid P'')& \mbox{if } P = ({\sf def}~ D ~{\sf in}~ (X(\tilde{v})\ |\ P'))\mid P''\\
& \hspace{5.5ex} \mbox{where } D::=X(\tilde{x}) = Q \\
\nextProc(P'\mid P'')& \mbox{if } P = ({\sf def}~ D ~{\sf in}~  P')\mid P''\\
& \hspace{5.5ex} \mbox{where } D::=X(\tilde{x}) = Q \\
&\mbox{ and } (X(\tilde{v})\ \notin P'\\
0 & \mbox{otherwise}
\end{cases} $ .}

\section{Global and Local Types}\label{section:multiparty_session}

Aiming to represent the available knowledge more accurately, we use 
imprecise probability given in terms of lower and upper probabilities to 
represent and quantify uncertainty in multiparty session types. Since in 
the processes presented in Section \ref{section:multiparty_processes} the 
probabilities are static, the global types are used to check whether the 
probabilities of executing actions are in the intervals indicated by the 
imprecise~probabilities.

The {\bf global types} $G,G',\ldots$ presented in Table \ref{table:global} 
describe the global behaviour of a probabilistic multiparty session 
process. We use the imprecise probability $\delta$ having the form 
$[d_1,d_2]$ in which $d_1,d_2 \in [0,1]$ and $d_1\leq d_2$. If 
$\delta=[d,d]$, we use the shorthand notation~$\delta=d$.

\begin{table}[ht]
  \centering
\begin{tabular}{l@{\hspace{3ex}}lll@{\hspace{2ex}}r}
   \hline
   \vspace{-1.5ex}
   $~$\\
   \vspace{0.5ex}
    {\it Global} & $G$ & $::=$ & $r_1 \rightarrow r_2: \{ \delta_i : l_i (S_i).G_i\}_{i \in I}$   & (interaction with $I \neq \emptyset$)\\
    \vspace{0.5ex}
    & & $\shortmid$ & $\mu t.G$ & (recursive with $G\neq t$)\\
        \vspace{0.5ex}
    & & $\shortmid$ & $t$ & (variable)\\
        \vspace{0.5ex}
    & & $\shortmid$ & {\End} & (end)\\
       \vspace{0.5ex}
    {\it Sorts} & $S$ & $::=$ & $\bool \mid \nat \mid \ldots$ & (base types)\\
    \hline
       \vspace{0.15ex}
   \end{tabular}
\caption{Global types of the session typing}\label{table:global}
\end{table}
\vspace{-2ex}

Type $r_1 \rightarrow r_2: \{ \delta_i : l_i (S_i).G_i\}_{i \in I}$ states 
that a participant with role $r_1$ sends with a probability belonging to 
the imprecise probability $\delta_i$ a message of type $S_i$ to a 
participant with role $r_2$ by using label $l_i$; then we have the 
interactions described by $G_i$. Each~$l_i$ is unique, as it was already 
assumed when defining the processes. Type $\mu t.G$ is recursive, 
where type variable $t$ is guarded in the standard way (it appears only 
under a prefix).

\begin{example} \label{example:global}
The following global type formalizes the $\PhoneNoA$ session of Example 
\ref{example:processes}:
\smallskip

\centerline{$G_A=\mu t.r_B \rightarrow r_A:
\begin{Bmatrix*}[l]
\delta_1: talk(string).r_A \rightarrow r_B:
\begin{Bmatrix*}[l]
\delta_2: yes(string).t,\\
\delta_3: no(string).t;\\
\delta_4: unsure(string).t;\\
\delta_5: quit(string).\End
\end{Bmatrix*},\\
\delta_6: quit(string).\End 
\end{Bmatrix*}$ .}

\smallskip

\noindent If for all $i\in \{1,\ldots 6\}$ it holds that $\delta_i=[0,1]$, 
then the probabilities in the processes used in 
Example~\ref{example:processes} fulfilling the global type $G_A$ are 
useless. This means that any probability used in the $\PhoneNoA$ session of 
Example~\ref{example:processes} does not affect the fulfilling of global 
type~$G_A$ (they are irrelevant for the global type). However, if 
$\delta_6=[0.95,1]$, then this implies that $\Bob$ (the person in charge of 
performing the poll) prefers to avoid talking with other persons in order 
to complete the poll (as required). This would lead to a sampling and 
coverage errors. This kind of behaviour should not be allowed by imposing 
$\delta_6=[0,0.1]$, meaning that the chance for $\Bob$ to stop the polling 
is minimal. To obtain the imprecise probabilities of the global types, we 
just need to subtract and add the error to a sample mean; these 
values can be obtained by using statistical methods. 
\end{example}
 
The {\bf local types} \; $T,T',\ldots$ presented in Table \ref{table:local} 
describe the local behaviour of processes; they also represent a 
connection between the global types and processes of our~calculus.

\begin{table}[ht]
  \centering
 \begin{tabular}{l@{\hspace{3ex}}lll@{\hspace{2ex}}r}
   \hline
      \vspace{-1ex}
   $~$\\
   \vspace{0.5ex}
    {\it Local} & $T$ & $::=$ & $r\ \oplus_{i\in I}\ \delta_i:\ !\, l_i\langle S_i \rangle.T_i$   & (selection towards role $p$, $I \neq\emptyset$) \\
    \vspace{0.5ex}
    & & \quad $\shortmid$ & $r\ \&_{i\in I}\ ?\, l_i(S_i).T_i $   & (branching from role $p$, $I \neq\emptyset$) \\
       \vspace{0.5ex}
    & & \quad $\shortmid$ & $\mu t.T$ & (recursive with $T \neq t$)\\
        \vspace{0.5ex}
    & & \quad $\shortmid$ & $t$ & (variable)\\
        \vspace{0.5ex}
    & & \quad $\shortmid$ & \End & (end)\\     
        \vspace{0.5ex}
    {\it Sorts} & $S$ & $::=$ & $\nat \mid \bool \mid \ldots$ & (base types)\\
    \hline 
        \vspace{0.15ex}
   \end{tabular}
\caption{Local types of the session typing}\label{table:local}
\end{table}
\vspace{-2ex}

The selection $r\ \oplus_{i\in I}\ \delta_i:\,!l_i\langle S_i \rangle.T_i$ 
describes a channel that can choose a label~$l_i$ with a probability 
belonging to the imprecise probability $\delta_i$ (for any $i \in I$), and 
send it to $r$ together with a variable of type $S_i$; then the channel 
must be used as described by $T_i$. The branching type $r\ \&_{i\in 
I}\,?l_i(S_i).T_i$ describes a channel that can receive a label~$l_i$ from 
role~$r$ (for some $i \in I$, chosen by $r$), together with a variable of 
type $S_i$; then the channel must be used as described in~$T_i$. The labels 
$l_i$ of selection and branching types are all distinct, and their order is 
irrelevant. The recursive type $\mu t.T$ and type variable $t$ model 
infinite behaviours. Type \End ~is the type of a terminated channel (and it 
is often omitted). The base types $S, S_0, S_1, \ldots$ can be types like 
$\bool$, $\Int$, etc. For simplicity, as done by Honda et al. 
in~\cite{Honda16}, the local types do not contain the parallel composition.

We define now the projection of a global type to a local type for each~participant.

\begin{definition}\label{definition:projection}
The projection $G\upharpoonright r$ for a participant with role $r$ appearing in a global type~$G$ is inductively defined~as:
\begin{itemize}
\item[$\bullet$] 
$(r_1 \!\rightarrow\! r_2\!:\! \{ \delta_i : l_i (S_i).G_i\}_{i \in I})\!\upharpoonright\! r = \begin{cases} 
r_1 \!\oplus_{i\in I}\! \delta_i:!l_i\langle S_i 
\rangle.(G_i \upharpoonright r ) & \mbox{if } r=r_1 \neq r_2\\
r_2\ \&_{i\in I}\ ?l_i(S_i).(G_i \upharpoonright r ) &\mbox{if } r=r_2 \neq r_1 \\
G_{1}\upharpoonright r & \mbox{if } r\neq r_1  \mbox{ and } r \neq r_2\\
& \forall i, j \in J,\  G_{i}\upharpoonright r=G_{j}\upharpoonright r 
\end{cases} $;

\item[$\bullet$] 
$(\mu t.G) \upharpoonright r = 
\begin{cases}
\mu t. (G \upharpoonright r) & \mbox{if } G \upharpoonright r \neq \End\\
\End
\end{cases} $;

\item[$\bullet$] $t \upharpoonright r = t$ ;

\item[$\bullet$] $\End \upharpoonright r =\End$ .
\end{itemize}
 
When none of the side conditions hold, the projection is undefined.
\end{definition}
In the global type $r_1 \rightarrow r_2: \{ \delta_i : l_i (S_i).G_i\}_{i \in I}$, 
the values of the imprecise probability $\delta_i=[d_{1i};d_{2i}]$ should 
not be arbitrarily assigned, but rather satisfy some restrictions as 
defined in what follows.
\vspace{0.5ex}

\begin{definition}
Consider a set of probabilities $\{\delta_i\}_{i \in I}$, where the 
imprecise probabilities have the form $\delta_i=[d_{1i};d_{2i}]$. The set 
$\{\delta_i\}_{i \in I}$ is called proper if

\centerline{$\sum_{i \in I} d_{1i} \leq 1 \leq \sum_{i \in I} d_{2i}$ ,}

\noindent while the set $\{\delta_i\}_{i \in I}$ is called reachable if 

\centerline{$(\sum_{i \neq j} d_{1j}) + d_{2i} \leq 1 \leq (\sum_{i \neq j} d_{2j}) +d_{1i}$ .}
\end{definition}

By defining a set to be proper and reachable, we avoid cases in which there 
does not exist $p_i \in \delta_i$ for all $i\in I$ such that $\sum_{i\in I} 
p_i=1$. More details about the imprecise probabilities and operations on 
them can be found in \cite{CamposHM94}.

From now on we assume that all the global types are well-formed, i.e. $G 
\upharpoonright r$ is defined for all roles $r$ occurring in $G$, and all 
the types with imprecise probabilities contain reachable sets of imprecise 
probabilities.
        \vspace{0.5ex}

\begin{example} \label{example:well-formed}
Consider the global type $G_A$ of Example \ref{example:global}, in which 
for all $i\in \{1,\ldots 5\}$ it holds that $\delta_i=[0,1]$, and also 
$\delta_6=[0.95,1]$. According to Definition \ref{definition:projection}, 
$G_A \upharpoonright r$ is defined for all roles. Moreover, according to the 
above definitions, the set $\{\delta_i\}_{i \in \{2,3,4,5\}}$ is proper and 
reachable, and the set $\{\delta_i\}_{i \in \{1,6\}}$ is proper but not 
reachable. 
\end{example}

\section{Probabilistic Multiparty Session Types}
\label{section:typing_system}

Here we show how to connect the probabilistic processes of
Section \ref{section:multiparty_processes} to the global and local 
types using imprecise probabilities presented in 
Section \ref{section:multiparty_session}. 
For this, we introduce sortings and typings with the purpose of 
defining types for probabilistic~behaviours:
        \vspace{0.5ex}

\centerline{$\Gamma$ $::=$ $\emptyset \,\mid\, \Gamma, x: S \,\mid\, \Gamma, 
X:\tilde{S}\tilde{T}$ 
\quad and \quad 
$\Delta$ $::=$ $\emptyset \,\mid\, \Delta, \{s[r_i]:T_i\}_{i\in I}$ .}

\noindent
The typing system uses a map from shared names to their sorts $(S,S', \ldots)$.
A sorting $(\Gamma,\Gamma', \ldots)$ is a finite map from names to sorts, 
and from process variables to sequences of sorts and types. A typing 
$(\Delta,\Delta',\ldots)$ records linear usage of session channels. Given 
two typings $\Delta$ and~$\Delta'$, their disjoint union is denoted by 
$\Delta,\Delta'$ (by assuming that their domains contain disjoint sets of 
session channels).

The type assignment system for processes is given in Table 
\ref{table:typing}, where $pid(G)$ denotes the set of participants in $G$.
We use the judgement $\Gamma \vdash P \rhd \Delta$ saying that ``under the 
sorting $\Gamma$, process $P$ has typing $\Delta$''. ({\sc TVar}) and 
({\sc TVal}) are the rules for typing variables and values. 
({\sc TSelect}) and ({\sc TBranch}) are the rules for typing selection and 
branching, respectively. As these rules contain probabilities; the rules 
should check all the possible choices with respect to~$\Gamma$. These two 
rules state that selection (branching) process is well-typed if~$c[r]$ has 
a compatible selection (branching) type, and the continuations $P_i$ (for 
all $i \in I$) are well-typed with respect to the session types. Rule ({\sc 
TConc}) composes two processes if their local types are disjoint. Rule 
({\sc TEnd}) is standard; ``$\Delta \ \End\ \only$'' means that $\Delta$ 
contains only \End \; as session types. ({\sc TRes}) is the restriction 
rule for session names and claims that $(\nu s)P$ is well-typed in 
$\Gamma$; this happens only if the types~$T_i$ of the session $s$ are 
exactly the projections of the same global type $G$ into all participants 
of~$G$. {\sc (TDef)} says that a process definition 
$\defin{X(\tilde{x}c_1\ldots c_n )=P}{Q}$ is well-typed if both~$P$ and $Q$ 
are well-typed in their typing contexts. Rule {\sc (TCall)} says that a 
process call $X(\tilde{v}c_1\ldots c_n)$ is well-typed if the actual 
parameters $\tilde{v}c_1\ldots c_n$ have compatible types with respect to $X$.

\begin{table*}[ht]
  \centering
   \begin{tabular}{@{\hspace{0ex}}c@{\hspace{-10ex}}r@{\hspace{0ex}}}
   \hline
\vspace{-1.5ex}
$~$\\
   \vspace{1.25ex}
   $\Gamma,x:S\vdash x:S$   \quad
   $\Gamma\vdash v:S$  
   & ({\sc TVar}), ({\sc TVal})
   \\
   \vspace{1.25ex}
   $\dfrac{{ \forall i.\Gamma \vdash v_i:S_i \quad \forall i.\Gamma\vdash P_i \rhd \Delta,c:T_i} \quad \sum_{i\in I} p_i=1 \quad p_i \in \delta_i}{ \Gamma \vdash  c[r] \oplus_{i\in I}\langle p_i:l_i(v_i); P_i\rangle\rhd \Delta,c:r\ \oplus_{i\in I}\ \delta_i:!l_i\langle S_i \rangle.T_i}$ & ({\sc TSelect})\\
     \vspace{1.25ex}
   $\dfrac{ \forall i.\Gamma,x_i:S_i\vdash P_{i} \rhd \Delta,c:T_i }{ \Gamma \vdash  c[r]\&_{i\in I}\{l_i(x_i); P_i\}\rhd \Delta, c: r\ \&_{i\in I}\ ?l_i(S_i).T_i}$ & ({\sc TBranch})\\   
    \vspace{1.25ex}
  $\dfrac{ \Delta~\End~\only}{ \Gamma\vdash \0\rhd\Delta}$ \qquad $\dfrac{ \Gamma\vdash P \rhd \Delta \qquad \Gamma\vdash Q \rhd \Delta'}{ \Gamma\vdash P \mid Q \rhd \Delta, \Delta'}$ & ({\sc TEnd}), ({\sc TConc})\\    
     \vspace{0.5ex}
 $\dfrac{ pid(G)=|I| \quad \forall i.G\upharpoonright r_i =T_i \quad \Gamma \vdash P \rhd \Delta,\{s[r_i]:T_i\}_{i \in I}}{ \Gamma\vdash (\nu s) P \rhd\Delta}$
&  ({\sc TRes})\\ 
   \vspace{1.25ex}
 $\dfrac{ 
\begin{array}{c}
 \Gamma, x:\tilde{S}, X:\tilde{S}T_1\ldots T_n\vdash P \rhd c_1:T_1,\ldots,c_n:T_n  \qquad
  \Gamma, X:\tilde{S}T_1\ldots T_n\vdash Q \rhd \Delta
  \end{array}}{ \Gamma\vdash \defin{ X(\tilde{x}c_1\ldots c_n )=P}{Q}\rhd\Delta}$&  ({\sc TDef})\\ 
           \vspace{1.25ex}
     $\dfrac{\Gamma \vdash \tilde{v}:\tilde{S}\qquad \Delta~\End~\only}{ \Gamma, X:\tilde{S}T_1\ldots T_n\vdash X(\tilde{v}c_1\ldots c_n) \rhd \Delta, c_1:T_1,\ldots,c_n:T_n}$ & ({\sc TCall})\\ 
    \hline 
     \vspace{0.15ex}
   \end{tabular}
\caption{Typing system for processes in probabilistic multiparty sessions}\label{table:typing}
\end{table*}
\vspace{-2ex}

As in \cite{Honda16}, an {\it annotated} process $P$ is the result of 
annotating the bound names of $P$, e.g. $(\nu s:G)P$ and $s?(x : S)P$. 
We consider that these annotations are natural for our framework. 
For typing annotated processes, we assume the obvious updates for the 
rules of Table~\ref{table:typing}. For instance, in the annotated rule 
obtained from rule {\sc (TRes)} for typing $(\nu s:G)P$, the set 
$\{s[r_i]:T_i\}_{i\in I}$ is obtained from projecting the type $G$ of $s$ 
such that $T_i$ is the projection of~$G$ onto~$r_i$ for all $i\in I$. It 
is worth mentioning that some rules which do not involve variables are the 
same as the ones of Table \ref{table:typing}.

\begin{theorem}\label{theorem:decidable}
Given an annotated process $P$ and a sorting $\Gamma$, it is {\it decidable} 
whether there is a typing $\Delta$ such that $\Gamma \vdash P \rhd \Delta$. 
If such a typing $\Delta$ exists, there is an algorithm to build~one.
\end{theorem}
\begin{proof}
The annotated rules obtained from the rules of Table \ref{table:typing} 
are used to construct the typing of $P$ under $\Gamma$. If the 
construction does not succeed at some point, the algorithm fails. We use 
the same notation $\Gamma \vdash P \rhd \Delta$ to denote the construction 
of $\Delta$ out of $P$ and~$\Gamma$. The algorithm aborts when a 
constraint in a rule is violated. Note that if there exists a derivation 
for $P$ under $\Gamma$ in the typing system, then the construction of the 
algorithm is possible. Therefore, the algorithm gives a decidable 
procedure for the typability of annotated processes. 
\end{proof}
\vspace{-2ex}

Since the processes interact, their dynamics is formalized as in 
\cite{Honda16} by a labelled type reduction relation $\xRightarrow{tl}$ 
on typing~$\Delta$ given by the following two rules:
\begin{itemize}
\item $s[r_1] : r_2\ \oplus_{i\in I}\ \delta_i:!l_i\langle S_i \rangle.T_i$, $s[r_2] : r_1\ \&_{i\in I}\ ?l_i(S_i).T'_i $ $\xRightarrow{(r_1,r_2,l_k)}_{\delta_{k}}$ $s[r_1]:T_k$, $s[r_2]:T'_k$ for $k\in I$;

\item $\Delta,\Delta' \xRightarrow{tl}_{\delta} \Delta,\Delta''$ \ whenever \ $\Delta' \xRightarrow{tl}_{\delta} \Delta''$.
\end{itemize}

\noindent 
The first rule corresponds to sending/receiving a value of type $S_k$ by 
using label~$l_k$ on session channel $s$ by the participant $r_2$. The 
second rule is used to compose typings when only a part of a typing is
changing. Just like for process, we need a normalization of the probability 
attached to the transition arrow in order to take into account the 
components of the typing. 
\smallskip

We are able to prove now that this probabilistic typing system is sound, 
namely its term checking rules admit only terms that are valid with respect 
to the structural congruence and operational semantics. {\it Subject 
reduction} ensures that the type of an expression is preserved during its 
evaluation. For the proof of the subject reduction, we need the following 
substitution and weakening~results.
     \vspace{0.5ex}

\begin{lemma} \label{lemma:subst_weak} \ 
\begin{itemize} 
\item (substitution) $\Gamma,x : S \vdash P \rhd \Delta$ and 
$\Gamma \vdash v : S$ \ imply \ $\Gamma \vdash P\{v/x\}  \rhd \Delta$. 

\item (type weakening) 
Whenever $\Gamma \vdash P \rhd \Delta$ is derivable, 
then its weakening is also derivable, \\ 
namely $\Gamma \vdash P \rhd \Delta,\Delta'$ for disjoint $\Delta'$, 
where $\Delta'$ contains only \End .

\item (sort weakening) If $X \notin dom(\Gamma)$ and $\Gamma \vdash P \rhd \Delta$, then $\Gamma, X:\tilde{S}\tilde{T} \vdash P \rhd \Delta$.

\item (sort strengthening) If $X \notin fpv(P)$ and $\Gamma, X:\tilde{S}\tilde{T} \vdash P \rhd \Delta$, then $\Gamma\vdash P \rhd \Delta$.
\end{itemize} 
\end{lemma}
\begin{proof}
The proofs are rather standard, following those presented in \cite{Honda16}.
\end{proof}

In our setting, at most one typing rule can be applied for any arbitrary 
given well-typed process. This is the reason why, by inverting a rule, we 
can describe how the (sub)processes of a well-typed process are typed. 
This is a basic property used in several papers when reasoning by induction 
on the structure of processes (for instance, \cite{Honda16,Coppo16,BocchiYY14}).

\begin{theorem}[type preservation under equivalence] \label{theorem:struct} \ \ 

\centerline{$\Gamma \vdash P\rhd \Delta$ and $P \equiv P'$ \ imply \ $\Gamma \vdash P'\rhd \Delta$ .}
\end{theorem}
\begin{proof}
The proof is by induction on $\equiv$ , showing (in both ways) that if one 
side has a typing, then the other side has the same typing. 
\begin{itemize}
\item Case $P \mid \0 \equiv P$ .

$\Rightarrow$ Assume $\Gamma \vdash P \mid \0\rhd \Delta$. By 
inverting the rule {\sc (TConc)}, we obtain $\Gamma \vdash P\rhd \Delta_1$ 
and $\Gamma \vdash \0 \rhd \Delta_2$, where $\Delta_1, \Delta_2 
=\Delta$. By inverting the rule {\sc (TEnd)}, $\Delta_2$ is only \End \; and 
$\Delta_2$ is such that $\dom(\Delta_1) \cap \dom(\Delta_2)=\emptyset$. 
Then, by type weakening, we get that $\Gamma \vdash P\rhd \Delta$, where 
$\Delta= \Delta_1,\Delta_2$.

$\Leftarrow$ Assume $\Gamma \vdash P\rhd \Delta$. By rule {\sc (TEnd)}, 
it holds that $\Gamma \vdash \0 \rhd \Delta'$, where $\Delta'$ is 
only {\End} and $\dom(\Delta) \cap \dom(\Delta')=\emptyset$. By 
applying the rule {\sc (TConc)}, we get $\Gamma \vdash P \mid \0\rhd 
\Delta,\Delta'$, and for $\Delta'=\emptyset$ we obtain $\Gamma \vdash P 
\mid \0\rhd \Delta$, as required.

\item Case $ (\nu s)\0 \equiv \0$.

$\Rightarrow$ Assume $\Gamma \vdash (\nu s)\0\rhd \Delta$. By inverting the 
rule {\sc (TRes)}, we get $\Gamma \vdash \0 \rhd \Delta,\{s[r_i]:T_i\}_{i 
\in I}$ where $pid(G)=|I|$ and $\forall i.G\upharpoonright r_i =T_i$. By 
inverting the rule {\sc (TEnd)}, $\Delta,\{s[r_i]:T_i\}_{i \in I}$ is only 
\End , namely $\Delta$ is only \End \;. Using the rule {\sc (TRes)}, we get 
$\Gamma \vdash \0\rhd \Delta$.

$\Leftarrow$ Assume $\Gamma \vdash \0\rhd \Delta$. By rule {\sc (TRes)}, it 
holds that $\Gamma \vdash (\nu s)\0\rhd \Delta_1$ where $pid(G)=|I|$, 
$\forall i.G\upharpoonright r_i =\End$, $\Delta_2=\{s[r_i]:\End\}_{i \in 
I}$ and $\Delta=\Delta_1,\Delta_2$. Then, by type weakening, we get that 
$\Gamma \vdash (\nu s)\0\rhd \Delta$, as required.

\item Case $P \mid Q \equiv Q \mid P$

$\Rightarrow$ Assume $\Gamma \vdash P \mid Q\rhd \Delta$. By inverting the 
rule {\sc (TConc)}, we obtain $\Gamma \vdash P\rhd \Delta_1$ and $\Gamma 
\vdash Q \rhd \Delta_2$, where $\Delta_1,\Delta_2 =\Delta$. Using the rule 
{\sc (TConc)}, we get $\Gamma \vdash Q \mid P\rhd \Delta$.

$\Leftarrow$ In a similar way (actually symmetric to $\Rightarrow$).

\item Case $(P \mid Q) \mid R \equiv P \mid (Q \mid R)$

$\Rightarrow$ Assume $\Gamma \vdash (P \mid Q) \mid R\rhd \Delta$. By 
inverting the rule {\sc (TConc)}, we obtain $\Gamma \vdash P\rhd \Delta_1$, 
$\Gamma \vdash Q \rhd \Delta_2$ \\ and $\Gamma \vdash R \rhd \Delta_3$, 
where $\Delta_1, \Delta_2, \Delta_3 =\Delta$. Using the rule {\sc (TConc)}, 
we get $\Gamma \vdash P \mid (Q \mid R)\rhd \Delta$.

$\Leftarrow$ In a similar way (actually symmetric to $\Rightarrow$).

\item Case $(\nu s)(\nu s')P \equiv (\nu s')(\nu s)P$

$\Rightarrow$ Assume $\Gamma \vdash (\nu s)(\nu s')P\rhd \Delta$. By 
inverting the rule {\sc (TRes)} twice, we obtain $\Gamma \vdash P\rhd 
\Delta,\{s[r_i]:T_i\}_{i \in I},\{s'[r'_i]:T'_i\}_{i \in I'}$ where 
$pid(G)=|I|$, $\forall i.G\upharpoonright r_i =T_i$, $pid(G')=|I'|$ and 
$\forall i.G'\upharpoonright r'_i =T'_i$. Using the rule {\sc (TRes)}, 
we get $\Gamma \vdash (\nu s')(\nu s)P\rhd \Delta$.

$\Leftarrow$ In a similar way (actually symmetric to $\Rightarrow$).

\item Case $(\nu s)P | Q \equiv (\nu s)(P | Q)~ ({\rm if}~ s \not\in fc(Q))$

$\Rightarrow$ Assume $\Gamma \vdash (\nu s)P | Q\rhd \Delta$. By inverting 
the rule {\sc (TConc)}, we obtain $\Gamma \vdash (\nu s)P\rhd \Delta_1$ and 
$\Gamma \vdash Q \rhd \Delta_2$, where $\Delta_1,\Delta_2 =\Delta$. By 
inverting the rule {\sc (TRes)}, we obtain $\Gamma \vdash P \rhd \Delta_1,\{s[r_i]:T_i\}_{i \in I}$ where $pid(G)=|I|$ and $\forall i.G\upharpoonright r_i =T_i$. Using rule {\sc (TConc)}, we get $\Gamma \vdash P \mid Q \rhd \Delta_1,\{s[r_i]:T_i\}_{i \in I},\Delta_2$. As $s \not\in fc(Q))$ it means that $\Delta_2$ does not contain types for $s$, and thus from $\Gamma \vdash P \mid Q \rhd \Delta_1,\Delta_2,\{s[r_i]:T_i\}_{i \in I}$, where $pid(G)=|I|$ and $\forall i.G\upharpoonright r_i =T_i$, by using rule {\sc (TRes)}, we get $\Gamma \vdash (\nu s)(P | Q)\rhd \Delta$, where $\Delta_1,\Delta_2 =\Delta$.

$\Leftarrow$ Assume $\Gamma \vdash (\nu s)(P | Q)\rhd \Delta$. By inverting rule {\sc (TRes)}, we obtain $\Gamma \vdash P | Q\rhd \Delta,\{s[r_i]:T_i\}_{i \in I}$, where $pid(G)=|I|$ and $\forall i.G\upharpoonright r_i =T_i$. Since $s \not\in fc(Q))$, by inverting 
the rule {\sc (TConc)}, we obtain $\Gamma \vdash P\rhd \Delta_1,\{s[r_i]:T_i\}_{i \in I}$ and 
$\Gamma \vdash Q \rhd \Delta_2$, where $\Delta_1,\Delta_2 =\Delta$. Using rule {\sc (TRes)}, we get $\Gamma \vdash (\nu s)P\rhd \Delta_1$. By using rule {\sc (TConc)}, we obtain $\Gamma \vdash (\nu s)P | Q\rhd \Delta$, as required.

\item Case $\defin{D}{\0} \equiv \0$

$\Rightarrow$ Assume $\Gamma \vdash \defin{D}{\0}\rhd \Delta$. By inverting 
rule {\sc (TDef)}, we obtain $\Gamma, x:\tilde{S}, X:\tilde{S}T_1\ldots 
T_n\vdash P \rhd c_1:T_1,\ldots,c_n:T_n$ and $\Gamma, X:\tilde{S}T_1\ldots 
T_n\vdash \0 \rhd \Delta$, where $D=\{X(\tilde{x}c_1\ldots c_n )=P\}$. By 
inverting rule {\sc (TEnd)} we get that $\Delta$ is only \End. Applying 
rule {\sc (TEnd)} we obtain $\Gamma\vdash \0 \rhd \Delta$.

$\Leftarrow$ Assume $\Gamma \vdash \0\rhd \Delta$. By inverting rule {\sc 
(TEnd)} we get that $\Delta$ is only \End. Applying rule {\sc (TEnd)} we 
obtain $\Gamma, X:\tilde{S}T_1\ldots T_n\vdash \0 \rhd \Delta$. Considering 
a process $P$ such that $\Gamma, x:\tilde{S}, X:\tilde{S}T_1\ldots 
T_n\vdash P \rhd c_1:T_1,\ldots,c_n:T_n$. By applying rule {\sc (TDef)} we 
get that $\Gamma \vdash \defin{X(\tilde{x}c_1\ldots c_n )=P}{\0}\rhd 
\Delta$, namely $\Gamma \vdash \defin{D}{\0}\rhd \Delta$, as required.

\item Case $\defin{D}{(\nu s)P} \equiv (\nu s)\defin{D}{P} ~({\rm if}~ s \not\in fc(D))$

$\Rightarrow$ Assume $\Gamma \vdash \defin{D}{(\nu s)P}\rhd \Delta$. By 
inverting rule {\sc (TDef)}, we obtain $\Gamma, x:\tilde{S}, 
X:\tilde{S}T_1\ldots T_n\vdash Q \rhd c_1:T_1,\ldots,c_n:T_n$ and $\Gamma, 
X:\tilde{S}T_1\ldots T_n\vdash (\nu s)P \rhd \Delta$, where 
$D=\{X(\tilde{x}c_1\ldots c_n )=Q\}$. By inverting the rule {\sc (TRes)}, 
we obtain $\Gamma \vdash P \rhd \Delta,\{s[r_i]:T_i\}_{i \in I}$ where 
$pid(G)=|I|$ and $\forall i.G\upharpoonright r_i =T_i$. Using rule {\sc 
(TDef)}, we get $\Gamma \vdash \defin{D}{P} \rhd \Delta,\{s[r_i]:T_i\}_{i 
\in I}$. using rule {\sc (TDef)}, we get $\Gamma \vdash (\nu 
s)\defin{D}{P}\rhd \Delta$.

$\Leftarrow$ Assume $\Gamma \vdash (\nu s)\defin{D}{P}\rhd \Delta$. By 
inverting rule {\sc (TRes)}, we obtain $\Gamma \vdash \defin{D}{P} \rhd 
\Delta,\{s[r_i]:T_i\}_{i \in I}$ where $pid(G)=|I|$ and $\forall 
i.G\upharpoonright r_i =T_i$. By inverting rule {\sc (TDef)}, we get 
$\Gamma, x:\tilde{S},X:\tilde{S}T_1\ldots T_n\vdash Q \rhd 
c_1:T_1,\ldots,c_n:T_n$ and $\Gamma, X:\tilde{S}T_1\ldots T_n\vdash P \rhd 
\Delta,\{s[r_i]:T_i\}_{i \in I}$, where $D=\{X(\tilde{x}c_1\ldots c_n)$=$Q\}$. 
Using rule {\sc (TRes)}, we get $\Gamma, X:\tilde{S}T_1\ldots T_n\vdash 
(\nu s)P \rhd \Delta$. By applying rule {\sc (TDef)} we get that $\Gamma 
\vdash \defin{X(\tilde{x}c_1\ldots c_n )=Q}{(\nu s)P}\rhd \Delta$, namely 
$\Gamma \vdash \defin{D}{(\nu s) P}\rhd \Delta$, as required.

\item Case $\defin{D}{(P | Q)} \equiv (\defin{D}{P}) | Q~ ({\rm if}~ dpv(D) \cap fpv(Q) = \emptyset)$

$\Rightarrow$ Assume $\Gamma \vdash \defin{D}{(P | Q)}\rhd \Delta$. By 
inverting rule {\sc (TDef)}, we get $\Gamma, x:\tilde{S},X:\tilde{S}T_1\ldots 
T_n\vdash R \rhd c_1:T_1,\ldots,c_n:T_n$ and $\Gamma,X:\tilde{S}T_1
\ldots T_n\vdash P\mid Q \rhd \Delta$, where $D=\{X(\tilde{x}c_1\ldots c_n)$
=$R\}$. By inverting the rule {\sc (TConc)}, 
we obtain $\Gamma, X:\tilde{S}T_1\ldots T_n\vdash P\rhd \Delta_1$ and 
$\Gamma, X:\tilde{S}T_1\ldots T_n\vdash Q\rhd \Delta_2$, where 
$\Delta_1,\Delta_2 =\Delta$. By applying rule {\sc (TDef)} we get that 
$\Gamma \vdash \defin{X(\tilde{x}c_1\ldots c_n )=R}{P}\rhd \Delta_1$, 
namely $\Gamma \vdash \defin{D}{P}\rhd \Delta_1$. As $dpv(D) \cap fpv(Q) = 
\emptyset$ and $X \in dpv(D)$, then $X \notin fpv(Q)$, and thus by sort 
strengthening we get $\Gamma \vdash Q\rhd \Delta_2$. By applying rule {\sc 
(TConc)} we get that $\Gamma \vdash (\defin{D}{P}) \mid Q\rhd \Delta$.

$\Leftarrow$ Assume $\Gamma \vdash (\defin{D}{P}) \mid Q\rhd \Delta$. By 
inverting rule {\sc (TConc)}, we obtain $\Gamma \vdash \defin{D}{P}\rhd 
\Delta_1$ and $\Gamma \vdash Q\rhd \Delta_2$, where $\Delta_1,\Delta_2 
=\Delta$. By inverting rule {\sc (TDef)}, we get $\Gamma, x:\tilde{S}, 
X:\tilde{S}T_1\ldots T_n\vdash R \rhd c_1:T_1,\ldots,c_n:T_n$ and $\Gamma, 
X:\tilde{S}T_1\ldots T_n\vdash P \rhd \Delta_1$, where 
$D=\{X(\tilde{x}c_1\ldots c_n )=R\}$. As $dpv(D) \cap fpv(Q) = \emptyset$ 
and $X \in dpv(D)$, then $X \notin fpv(Q)$, namely $X \notin dom(\Gamma)$, 
and thus by sort weakening we get $\Gamma, X:\tilde{S}T_1\ldots T_n\vdash Q 
\rhd \Delta_2$.  By applying rule {\sc (TConc)} we get that $\Gamma, 
X:\tilde{S}T_1\ldots T_n \vdash P \mid Q\rhd \Delta$. By using rule {\sc 
(TDef)}, we get $\Gamma \vdash \defin{X(\tilde{x}c_1\ldots c_n )=R}{(P | 
Q)}\rhd \Delta$, namely $\Gamma \vdash \defin{D}{(P | Q)}\rhd \Delta$, as 
required.
 
\item Case $\defin{D}{\defin{D'}{P}} \equiv \defin{D'}{\defin{D}{P}}$ $({\rm if}~ (dpv(D) \cup fpv(D)) \cap dpv(D') = (dpv(D') \cup fpv(D')) \cap dpv(D) = \emptyset)$

$\Rightarrow$ Assume $\Gamma \vdash \defin{D}{\defin{D'}{P}}\rhd \Delta$. 
By inverting rule {\sc (TDef)}, we get $\Gamma, x:\tilde{S}, 
X:\tilde{S}T_1\ldots T_n\vdash R \rhd c_1:T_1,\ldots,c_n:T_n$ and $\Gamma, 
X:\tilde{S}T_1\ldots T_n\vdash \defin{D'}{P} \rhd \Delta$, where 
$D=\{X(\tilde{x}c_1\ldots c_n )=R\}$. By inverting rule {\sc (TDef)}, we 
obtain $\Gamma,, X:\tilde{S}T_1\ldots T_n, x':\tilde{S'}, 
X':\tilde{S'}T'_1\ldots T'_n\vdash R' \rhd c'_1:T'_1,\ldots,c'_n:T'_n$ and 
$\Gamma, X:\tilde{S}T_1\ldots T_n, X':\tilde{S'}T'_1\ldots T'_n\vdash P 
\rhd \Delta$, where $D'=\{X(\tilde{x'}c'_1\ldots c'_n )=R'\}$. As $(dpv(D) 
\cup fpv(D)) \cap dpv(D') = (dpv(D') \cup fpv(D')) \cap dpv(D) = 
\emptyset$, it means that $X$ and $X'$ are different. Applying rule rule 
{\sc (TDef)} twice, we get $\defin{D'}{\defin{D}{P}}$.

$\Leftarrow$ In a similar way (actually symmetric to $\Rightarrow$).
\end{itemize}
\end{proof}
The following result is known also as `Subject Reduction'. According 
to this result, if a well-typed process takes a transition step 
of any kind, the resulting process is also well-typed.

\begin{theorem}[type preservation under evolution]\label{theorem:reduction} \ \\
\centerline{$\Gamma \vdash P\rhd \Delta$ and $P\xrightarrow{tl}_{p} P'$ imply 
$\Gamma \vdash P'\rhd \Delta'$, where $\Delta'$=$\Delta$ or 
$\Delta \xRightarrow{tl}_{\delta} \Delta'$ with~$p \cdot \nextProc(P) \in \delta$.}
\end{theorem}
\begin{proof}
By induction on the derivation of $P \xrightarrow{tl}_{p_i} P'$. There is a 
case for each operational semantics rule, and for each operational 
semantics rule we consider the typing system rule generating $\Gamma 
\vdash P \rhd \Delta$.
\begin{itemize}
\item Case {\sc (Com)}: {$ s[r_1][r_2] \oplus_{j\in J}\langle p_j:l_j(v_j); P_j\rangle  \mid s[r_2][r_1]\&_{i\in I}\{l_i(x_i); P'_i\} \xrightarrow{(r_1,r_2,l_k)}_{p_k} P_k \mid  P'_k\{\tilde{v_k}/\tilde{x_k}\}$.}

By assumption, $\Gamma \vdash  s[r_1][r_2] \oplus_{j\in J}\langle 
p_j:l_j(v_j); P_j\rangle \mid s[r_2][r_1]\&_{i\in I}\{l_i(x_i); P'_i\} \rhd 
\Delta$. By inverting the rule {\sc (TConc)}, we get $\Gamma \vdash 
 s[r_1][r_2] \oplus_{j\in J}\langle p_j:l_j(v_j); P_j\rangle \rhd 
\Delta_1 $ and $\Gamma \vdash  s[r_2][r_1]\&_{i\in I}\{l_i(x_i); 
P'_i\} \rhd \Delta_2$ with $\Delta=\Delta_1, \Delta_2$. Since these can be 
inferred only from {\sc (TSelect)} and {\sc (TBranch)}, we know that 
$\Delta_1=\Delta'_1,s[r_1]:r_2\ \oplus_{i\in I}\ \delta_i:!l_i\langle S_i 
\rangle.T_i$ and $\Delta_2=\Delta'_2,s[r_2]: r_1\ \&_{i\in I}\ 
?l_i(S_i).T'_i$. By inverting the rules {\sc (TSelect)} and {\sc 
(TBranch)}, we get that $\forall i.\Gamma \vdash v_i:S_i$, $\forall 
i.\Gamma\vdash P_i \rhd \Delta'_1,s[r_1]:T_i$, $\sum_{i\in I} p_i=1$, $p_i 
\in \delta_i$ and $\forall i.\Gamma,x_i:S_i\vdash P'_{i} \rhd 
\Delta'_2,s[r_2]:T'_i$. From $\Gamma \vdash v_i:S_i$ and 
$\Gamma,x_i:S_i\vdash P'_{i} \rhd \Delta'_2,s[r_2]:T'_i$, by applying the 
substitution part of Lemma~\ref{lemma:subst_weak}, we get that $\Gamma 
\vdash P'_i\{v_i/x_i\} \rhd \Delta'_2,s[r_2]:T'_i$. By applying the rule 
{\sc (TConc)}, we obtain $\Gamma \vdash P_i \mid P_i\{v_i/x_i\} \rhd 
\Delta'_1,s[r_1]:T_i,\Delta'_2,s[r_2]:T'_i$. By using the type reduction 
relation, we get that $\Delta \xRightarrow{(r_1,r_2,l_k)}_{\delta_k} \Delta'$, 
where $\Delta'=\Delta'_1,s[r_1]:T_k,\Delta'_2,s[r_2]:T'_k$ and $p_k \in \delta_k$. 

\item Case {\sc (Call)}: ${\sf def}~ X(\tilde{x}) = R ~{\sf in}~ 
(X(\tilde{v})\ |\ Q) \xrightarrow{\varepsilon}_1 {\sf def}~ X(\tilde{x}) = 
R ~{\sf in}~ (R\{\tilde{v}/\tilde{x}\}\ |\ Q)$

By assumption, $\Gamma \vdash {\sf def}~ X(\tilde{x}) = R ~{\sf in}~ 
(X(\tilde{v})\ |\ Q) \rhd \Delta$. By inverting rule {\sc (TDef)}, we 
obtain $\Gamma, x:\tilde{S}, X:\tilde{S}\vdash R \rhd \emptyset$ and 
$\Gamma, X:\tilde{S}\vdash X(\tilde{v})\ |\ Q \rhd \Delta$. By inverting 
the rule {\sc (TConc)}, we obtain $\Gamma, X:\tilde{S} \vdash 
X(\tilde{v})\rhd \Delta_1$ and $\Gamma, X:\tilde{S} \vdash Q \rhd 
\Delta_2$, where $\Delta_1,\Delta_2 =\Delta$. By inverting the rule {\sc 
(TCall)}, we get $\Gamma \vdash \tilde{v} : \tilde{S}$ and $\Delta_1$ is 
\End~only. Applying substitution (Lemma \ref{lemma:subst_weak}), we get 
$\Gamma, X:\tilde{S}\vdash R\{\tilde{v}/\tilde{x}\} \rhd \emptyset$. As 
$\Delta_1$ is \End~only, then by type weakening (Lemma 
\ref{lemma:subst_weak}) we get that $\Gamma, X:\tilde{S}\vdash 
R\{\tilde{v}/\tilde{x}\} \rhd \Delta_1$. Applying rule {\sc (TConc)}, we 
get $\Gamma, X:\tilde{S}\vdash R\{\tilde{v}/\tilde{x}\} |\ Q \rhd \Delta$. 
Applying rule {\sc (TDef)}, we obtain $\Gamma \vdash {\sf def}~ 
X(\tilde{x}) = R ~{\sf in}~ (R\{\tilde{v}/\tilde{x}\}\ |\ Q) \rhd \Delta$, 
as desired.

\item Case {\sc (Ctxt)}: $P \xrightarrow{tl}_p P'$ implies $\E[P] 
\xrightarrow{tl}_{p'} \E[P']$, where $p'=p \cdot \nextProc(P)/ 
\nextProc(\E(P))$

Based on the form of the context we got several cases.
\begin{itemize}
\item Case $\E=(\nu s) [\;]$. By assumption, $\Gamma \vdash (\nu s)P \rhd 
\Delta$. By inversing the rule {\sc (TRes)}, we get that $\Gamma \vdash P 
\rhd \Delta,\{s[r_i]:T_i\}_{i \in I}$, where $pid(G)=|I|$ and $\forall 
i.G\upharpoonright r_i =T_i$. By induction, $\Gamma \vdash P' \rhd 
\Delta',\{s[r_i]:T_i\}_{i \in I}$, where $\Delta'=\Delta$ or 
$\Delta\Rightarrow_{\delta} \Delta'$ and $p \cdot \nextProc(P)\in \delta$. 
By using the rule {\sc (TRes)}, we obtain $\Gamma \vdash (\nu s)P' \rhd 
\Delta'$, where $\Delta'=\Delta$ or $\Delta\Rightarrow_{\delta} \Delta'$, 
and $p' \cdot \nextProc((\nu s)P)= p \cdot \nextProc(P) \in \delta$, as 
required.

\item Case $\E=\defin{D}{[\;]}$. By assumption, $\Gamma \vdash \defin{D}{P} 
\rhd \Delta$. By inverting rule {\sc (TDef)}, we obtain $\Gamma, 
x:\tilde{S}, X:\tilde{S}T_1\ldots T_n\vdash Q \rhd c_1:T_1,\ldots,c_n:T_n$ 
and $\Gamma, X:\tilde{S}T_1\ldots T_n\vdash P \rhd \Delta$, where 
$D=\{X(\tilde{x}c_1\ldots c_n )=Q\}$. By induction, $\Gamma, 
X:\tilde{S}T_1\ldots T_n\vdash P' \rhd \Delta'$, where $\Delta'=\Delta$ or 
$\Delta\Rightarrow_{\delta} \Delta'$ and $p \cdot \nextProc(P)\in \delta$. 
By using the rule {\sc (TDef)}, we obtain $\Gamma \vdash \defin{D}{P'} \rhd 
\Delta'$, where $\Delta'=\Delta$ or $\Delta\Rightarrow_{\delta} \Delta'$, 
and $p' \cdot \nextProc(\defin{D}{P})= p \cdot \nextProc(P) \in \delta$, as 
required.

\item Case $\E=[\;] \mid P''$. By assumption, $\Gamma \vdash P \mid P'' 
\rhd \Delta$. By inverting rule {\sc (TConc)}, we obtain $\Gamma \vdash 
P\rhd \Delta_1$ and $\Gamma \vdash P''\rhd \Delta_2$, where 
$\Delta_1,\Delta_2 =\Delta$. By induction, $\Gamma\vdash P' \rhd 
\Delta'_1$, where $\Delta'_1=\Delta_1$ or $\Delta_1\Rightarrow_{\delta} 
\Delta'_1$ and $p \cdot \nextProc(P)\in \delta$. By using the rule {\sc 
(TConc)}, we obtain $\Gamma \vdash P' \mid P'' \rhd \Delta'$, where 
$\Delta'=\Delta$ or $\Delta\Rightarrow_{\delta} 
\Delta'=\Delta'_1,\Delta_2$, and $p' \cdot \nextProc(P\mid P'')= p \cdot 
\nextProc(P) \in \delta$, as required.

\end{itemize}

\item Case {\sc (Struct)}: $P\equiv P'$ and $P' \xrightarrow{tl}_r Q'$ and 
$Q \equiv Q' $ \ implies \ $ P \xrightarrow{tl}_r Q$.

It is obtained by using the structural congruence (Theorem~\ref{theorem:struct}).
\end{itemize} 
\end{proof}
As in the most approaches on multiparty session types, our approach 
ensures that a typed ensemble of processes interacting on a single 
annotated session (i.e. a typed $(\nu s:G)\mid_{i\in I} P_i)$ with each 
$P_i$ interacting only on~$s[r_i]$) is deadlock-free. The deadlock freedom 
property in the presence of multiple interleaved sessions was also studied 
in~\cite{Coppo16}. In what follows, $Q \not\rightarrow$ means that the 
process $Q$ is not able to evolve by means of any rule (we say that $Q$ is 
a stuck process). Negative premises are used to denote the fact that the 
passing to a new step is performed based on the absence of actions; the 
use of negative premises in this way does not lead to any inconsistency. 
\vspace{1ex}

\begin{theorem}[deadlock freedom]
Let $\emptyset \vdash P \rhd \emptyset$, where $P \equiv (\nu s:G) 
\mid_{i\in I} P_i$ and each~$P_i$ interacts only on $s[r_i]$ of type $T_i$, 
where $T_i =G \upharpoonright r_i$. Then $P$ is deadlock-free: i.e., $P 
\xrightarrow{tl}^*_p P'\not\rightarrow$ implies $P' \equiv \0$. 
\end{theorem}
\begin{proof}
Since the probabilities do not influence the behaviour,
the proof is similar to the approach presented in~\cite{Scalas17}. 
\end{proof}
\vspace{-1ex}

Lock-freedom was introduced in \cite{Kobayashi02} as a property stronger 
than deadlock freedom, requiring that every action that can be executed 
will eventually be executed. Session types not only ensure 
deadlock-freedom, but also lock-freedom~\cite{Scalas19}. Checking if such 
results still hold in our framework represents future work.

As it can be noticed from the rules of Table \ref{table:typing}, the 
obtained types are not unique. This is due to the fact that we use 
imprecise probabilities allowing the processes to be considered well-typed 
within the interval defined by lower and upper probabilities. Having 
several types for the same process, we can obtain some refinements of these 
typings. Inspired by \cite{BocchiYY14}, we use an \erase\ function that, 
when applied to a type, removes its probability annotations. Considering 
$\erase(\Delta_1)=\erase(\Delta_2)$, let us define the intersection of the 
imprecise probabilities appearing in the typing of processes as follows:

\centerline{$\Delta_1 \cap_{r_1} \Delta_2= \{s[r_2]: (T_1 \cap_p T_2)  \mid s[r_2]: T_1 \in \Delta_1  \mbox{ and } s[r_2]: T_2 \in \Delta_2\}$}

\noindent
such that
 
\centerline{$T_1 \cap_{r_1} T_2=$ $\begin{cases}
r_1\ \oplus_{i\in I}\ (\delta_{1i}\cap\delta_{2i}):\ !l_i\langle S_i \rangle.(T_{1i} \cap_{r_1} T_{2i}) & \mbox{if } T_{j}= r_1\ \oplus_{i\in I}\ \delta_{ji}:\ !l_i\langle S_i \rangle.T_{ji}\\
& \delta_{1i}\cap\delta_{2i} \neq \emptyset, \mbox{ with } j \in \{1,2\} \\
r_1\ \&_{i\in I}\ ?l_i(S_i).(T_{1i}\cap_{r_1} T_{2i}) & \mbox{if } T_j= r_1\ \&_{i\in I}\ ?l_i(S_i).T_{ji} \\
& \mbox{ with } j \in \{1,2\} \\
\mu t.(T_{1i} \cap_{r_1} T_{2i}) & \mbox{if } T_1= \mu t.T_{1i}  \mbox{ and } T_2=\mu t . T_{2i} \\
T_1 & \mbox{if } T_1=T_2=t   \mbox{ or } T_1=T_2=\End\\
{\it undefined} & \mbox{otherwise} .
\end{cases} $.}

In a similar manner, it can be defined the intersection 
$\Gamma_1 \cap_p \Gamma_2$ of two sortings. 
\smallskip

\noindent
For the next three results we assume that 
$\erase(\Delta_1)=\erase(\Delta_2)$ and 
$\erase(\Gamma_1)=\erase(\Gamma_2)$.
It is worth noting that a process well-typed with different types
still remains well-typed if we consider the intersections of all 
corresponding imprecise probabilities.
\vspace{1ex}

\begin{theorem}
If $\Gamma_1 \vdash P\rhd \Delta_1$ and $\Gamma_2 \vdash P\rhd \Delta_2$, 
then $\Gamma_1 \cap_r \Gamma_2 \vdash P\rhd \Delta_1 \cap_r \Delta_2$.
\end{theorem}
\begin{proof}
The proof follows by induction on the structure of process~$P$ and by using 
the rules of Table~\ref{table:typing}. 
\end{proof}

In what follows we associate to any execution path (starting from a 
process~$P$ and leading to a process $Q$) an execution probability of 
going along this path. This is computed by using the operational 
semantics presented in Table \ref{table:semantics}.
\vspace{0.5ex}

\begin{definition}[evolution paths]\label{definition:evolpath}
A sequence of evolution transitions represents an evolution path $ep=P_{0} 
\xrightarrow{tl_{1}}_{p_{1}} P_{1} \ldots \xrightarrow{tl_{k}}_{p_{k}} 
P_{k}$. Two evolution paths $ep=P_{0} \xrightarrow{tl_{1}}_{p_{1}} P_{1} 
\ldots \xrightarrow{tl_{k}}_{p_{k}} P_{k}$ and $ep'=P'_{0} 
\xrightarrow{tl'_{1}}_{p'_{1}} P'_{1} \ldots 
\xrightarrow{tl'_{k'}}_{p'_{k'}} P'_{k'}$ are identical if $k=k'$, and 
for all $t\in \{0,\ldots,k\}$ we have $P_i \equiv P'_i$, $p_i=p'_i$ 
and $tl_i=tl'_i$. 
\end{definition}
\vspace{0.5ex}

\begin{definition}[evolution probability]\label{definition:evolprob}
The probability to reach a process $P_j$ starting from a process $P_i$ 
(with $i<j$) along the evolution path $ep=P_{0} 
\xrightarrow{tl_{1}}_{p_{1}} P_{1} \ldots \xrightarrow{tl_{k}}_{p_{k}} 
P_{k}$ is denoted by $prob(ep,P_i,P_j)=p_{i+1} * \ldots * p_{j}$. The 
evolution probability to reach $P_j$ starting from a process~$P_i$ (with 
$i<j$) taking all the possible evolution paths (without considering 
identical paths) is denoted by $\prob(P_i,P_j)=\sum_{ep}prob(ep,P_i,P_j)$. 
\end{definition}

For any process $P$ we can define the sets $\Reach_k(P)$ of processes 
reached in $k$ steps starting from $P$ by using the rules of Table 
\ref{table:semantics}.

\begin{definition}[reachable sets]
$\Reach_1(P) =\{Q \mid P \xrightarrow{tl}_p Q\}$. \ For $k\geq 2$, 

\centerline{$\Reach_k(P) = \{Q \mid Q \in Reach_{k-1}(P) {\rm ~and ~} Q \not\rightarrow\} \cup \{R \mid  Q \in \Reach_{k-1}(P) {\rm ~and ~}  Q \xrightarrow{tl}_p R \}$.}

\end{definition}

\noindent 
The following result is based on the fact that for any process there exist 
various possible transitions, each occurring with a given probability. The 
approach is solid if for each well-typed process, the sum of these probabilities 
is~$1$.

\begin{theorem}\label{theorem:reach}
Given a well-typed process $P$ such that $\Reach_k(P) \neq \emptyset$ for 
$k\geq 1$, then

\centerline{${\sum_{Q\in \Reach_k(P)}prob(P,Q)=1}$ .}
\end{theorem}
\begin{proof}
By induction on the number $k$ of the evolution steps.
\begin{itemize}
\item Case $k=1$.

By induction on the structure of $P$, namely on the number of parallel components.
\begin{itemize}
\item Case $P=P_1$. We have several subcases:
\begin{itemize}

\item $P=\End$ or $P= c[r] \oplus_{i\in I}\langle 
p_i:l_i(v_i); P_i\rangle$ or $P= c[r]\&_{i\in I}\{l_i(x_i); 
P_i\}$ or $P= (\nu s)P_1$, meaning that no rule is applicable. Hence 
$\Reach_1(P)=\emptyset$, and the sum is not computed.

\item $P= {\sf def}~ X(\tilde{x}) = P' ~{\sf in}~ (X(\tilde{v}) | Q)$ meaning 
that the rule ({\sc Call}) is applied, and~$P$ evolves with probability $1$ to ${\sf def}~ X(\tilde{x}) = P' ~{\sf in}~ (P'\{\tilde{v}/\tilde{v}\} | Q)$. This means that 
${\sum_{Q\in \Reach_1(P)}\prob(P,Q)=1}$ (as desired).
\end{itemize}
\item Case $P=P_1 \mid P_2$. We have several subcases:
\begin{itemize}
\item $P_1=s[r_1][r_2] \oplus_{j\in J}\langle p_j:l_j(v_j); P_j\rangle$ and 
$P_2=s[r_2][r_1]\&_{i\in I}\{l_i(x_i); P'_i\}$ meaning that the rule ({\sc 
Com}) is applied, and $P$ evolves with probability $p_k$ (with $k\in J$) 
to $P_k \mid P'_k\{\tilde{v_k}/\tilde{x_k}\}$. Since the process $P$ is 
well-typed, it follows that ${\sum_{k\in J}p_k=1}$. Thus, ${\sum_{Q\in 
\Reach_1(P)}\prob(P,Q)=\sum_{k\in J}p_k=1}$ (as desired).
\end{itemize}
\end{itemize}
The remaining cases are proved in a similar manner.

\item Case $k>1$. 
Since $\Reach_k(P)\neq\emptyset$, this means that 
$\Reach_{k-1}(P)\neq\emptyset$. By considering the inductive step, we 
obtain that ${\sum_{Q\in \Reach_{k-1}(P)}\prob(P,Q)=1}$. Let us consider 
that another evolution step is performed, namely $k$ steps starting from 
process $P$. This means that ${\sum_{Q\in \Reach_k(P)}\prob(P,Q)}$ can be 
broken into two parts: one computing the probability for the first ($k-1$) 
steps, and another one using the probabilities obtained in the additional 
step. Thus it holds that ${\sum_{Q\in \Reach_k(P)}\prob(P,Q)}$
  $=\sum_{Q_1\in 
\Reach_{k-1}(P)}(\prob(P,Q_1) *$\linebreak 
$\sum_{Q\in \Reach_1(Q_1)}prob(Q_1,Q))$. 
Since we have $ \sum_{Q\in \Reach_1(Q_1)}prob(Q_1,Q)=1$, then \linebreak ${\sum_{Q\in \Reach_k(P)}\prob(P,Q)}$ 
  ${=\sum_{Q_1\in 
\Reach_{k-1}P)}\prob(P,Q_1) }$. Also, since we have 
\linebreak ${\sum_{Q_1\in \Reach_{k-1}(P)}\prob(P,Q_1)=1}$, \ 
then ${\sum_{Q\in \Reach_k(P)}\prob(P,Q)}=1$, as required.
\end{itemize}
\end{proof}

\section{Conclusion and Related Work}\label{section:conclusion} 

Aiming to represent and quantify uncertainty, imprecise probability 
generalizes probability theory to allow for partial probability 
specifications applicable when a unique probability distribution is 
difficult to be identified. In this paper we have used imprecise 
probability by introducing imprecise probabilities for multiparty session 
types in process algebra. According to our knowledge, there is no work in 
computer science dealing (explicitly) with imprecise probabilities.

We use a probabilistic extension of the process calculus presented 
in~\cite{Scalas17} by allowing both nondeterministic external choices and 
probabilistic internal choices. We define probabilistic multiparty session 
types able to codify the structure of the communications by using imprecise 
probabilities given in terms of lower and upper probabilities. Moreover, we 
have defined and studied a typing system extending the multiparty session 
types with imprecise probabilities; this typing system in which the 
channels have specific roles in a session is inspired by the system 
presented in \cite{Coppo16}. The new typing system has several properties 
and features. It preserves the classical typing properties; additionally, 
it is specified in such a way to satisfy the axioms of the standard 
probability theory and some properties of the imprecise probability theory. 
To illustrate our calculus, we considered an example inspired by 
probabilistic survey polls in which the probabilities are provided by a 
well-defined numerical scale for responses.

In \cite{Varacca07} there are proposed two semantics of a probabilistic 
variant of the $\pi$-calculus. For these, the types are used to identify a 
class of nondeterministic probabilistic behaviours which can preserve the 
compositionality of the parallel operator in the framework of event 
structures. The authors claim to perform an initial step towards a good 
typing discipline for probabilistic name passing by employing Segala 
automata \cite{Segala95} and probabilistic event structures. In comparison 
with them, we simplify the approach and work directly with processes, 
giving a probabilistic typing in the context of multiparty session types. 
Moreover, we include the imprecise probabilities in the framework of 
multiparty session types.

In a different framework presented in \cite{Llerena18}, the authors study  
the long-run behaviour of probabilistic models in the presence of 
uncertainty given by lower and upper bounds.

Several formal tools have been proposed for probabilistic reasoning. Some 
of these approaches make use of certain probabilistic logics. In 
\cite{Cooper14}, terms are assigned probabilistically to types via 
probabilistic type judgements, and from an intuitionistic typing system is 
derived a probabilistic logic as a subsytem \cite{Warrell16}. However, the 
existing probabilistic tools do not account for the imprecision of the 
probabilistic parameters in the model.


\bibliographystyle{unsrt}
\bibliography{referencesImprecise} 
 

\end{document}